\documentclass{article}
\usepackage[utf8]{inputenc}
\usepackage{amsmath,amssymb,amsthm,pxfonts,listings,color}
\usepackage[colorlinks]{hyperref}

\usepackage{mathrsfs}
\usepackage{stmaryrd}

\usepackage{etex}

\usepackage{amssymb}
\usepackage{amsmath}

\usepackage{mathrsfs}
\usepackage{mathtools}

\usepackage{stmaryrd}

\usepackage{verbatim}
\usepackage{url}
\usepackage{booktabs}
\usepackage{hyperref}

\usepackage{pgfplots}

\usepackage{tikz}
\usetikzlibrary{backgrounds, snakes,shapes.multipart,chains}
\tikzset{
    extra padding/.style={
        show background rectangle,
        inner frame sep=#1,
        background rectangle/.style={
            draw=none
        }
    },
    extra padding/.default=0.5cm
}
\tikzset{
  token/.style={
    rectangle split,
    rectangle split parts=4,
    rectangle split part fill=none,
    rectangle split ignore empty parts,
    draw=none}}

\makeatletter
\tikzset{north left/.code =\tikz@lib@place@handle@{#1}{north east}{-1}{0}{north west}{1}}
\tikzset{north right/.code=\tikz@lib@place@handle@{#1}{north west}{1}{0}{north east}{1}}
\tikzset{south left/.code =\tikz@lib@place@handle@{#1}{south east}{-1}{0}{south west}{1}}
\tikzset{south right/.code=\tikz@lib@place@handle@{#1}{south west}{1}{0}{south east}{1}}

\tikzset{west above/.code=\tikz@lib@place@handle@{#1}{south west}{0}{1}{north west}{1}}
\tikzset{west below/.code=\tikz@lib@place@handle@{#1}{north west}{0}{-1}{south west}{1}}
\tikzset{east above/.code=\tikz@lib@place@handle@{#1}{south east}{0}{1}{north east}{1}}
\tikzset{east below/.code=\tikz@lib@place@handle@{#1}{north east}{0}{-1}{south east}{1}}
\makeatother

\pgfplotsset{width=4cm}

\usepackage[all]{xy}
\usepackage{enumerate}
\usepackage{subcaption}
\usepackage[rightcaption]{sidecap}

\usepackage{syntax}
\usepackage{fancyvrb}
\usepackage{algpseudocode}
\usepackage{algorithm}

\usepackage{todo}

\algrenewcommand\algorithmicindent{1.0em}

\newtheorem{theorem}{Theorem}[section]

\newtheorem{remark}[theorem]{Remark}
 
\newtheorem{corollary}[theorem]{Corollary}
\newtheorem{lemma}[theorem]{Lemma}
\newtheorem{proposition}[theorem]{Proposition}

\theoremstyle{definition}

\newtheorem{definition}[theorem]{Definition}

\newcommand{\pair}[2]{\mbox{$\langle #1,\; #2 \rangle$}}
\newcommand{\tuple}[1]{\mbox{$\langle #1 \rangle$}}
\newcommand{\powerset}[1]{\mbox{$\mathcal{P}(#1)$}}

\newcommand{\set}[1]{\mbox{$\{ #1 \}$}}
\newcommand{\alt}{\mathrel{|}}

\newcommand{\imp}{\Rightarrow}
\newcommand{\bimp}{\Leftrightarrow}
\newcommand{\deq}{\mathrel{\mathop:}=}
\newcommand{\ptag}[1]{\tag*{\{#1\}}}

\newcommand{\nats}{\mathbb{N}}

\newcommand{\sP}{\mathsf{P}}

\newcommand{\cP}{\mathcal{P}}
\newcommand{\cQ}{\mathcal{Q}}

\newcommand{\cU}{\mathcal{U}}
\newcommand{\cV}{\mathcal{V}}

\newcommand{\cX}{\mathcal{X}}
\newcommand{\cY}{\mathcal{Y}}
\newcommand{\cZ}{\mathcal{Z}}

\newcommand{\mfrL}{\mathfrak{L}}

\newcommand{\mfrS}{\mathfrak{S}}

\newcommand{\bbB}{\mathbb{B}}

\newcommand{\bbP}{\mathbb{P}}

\newcommand{\defn}[1]{\textbf{#1}}

\usepackage{bussproofs}
  {\gdef\scalefactor{#1}\begin{center}\proofSkipAmount \leavevmode}%
  {\scalebox{\scalefactor}{\DisplayProof}\proofSkipAmount \end{center} }

\date{\today}

\title{Concurrent Kleene Algebra of Partial Strings}
\date{\today}

\author{Alex Horn and Jade Alglave}

\begin{document}

\maketitle

\begin{abstract}
Concurrent Kleene Algebra (CKA) is a recently proposed algebraic structure by Hoare and collaborators that unifies the laws of concurrent programming. The unifying power of CKA rests largely on the so-called exchange law that describes how concurrent and sequential composition operators can be interchanged. Based on extensive theoretical work on true concurrency in the past, this paper extends Gischer's pomset model with least fixed point operators and formalizes the program refinement relation by \'{E}sik's monotonic bijective morphisms to construct a partial order model of CKA. The existence of such a model is relevant when we want to prove and disprove properties about concurrent programs with loops. In particular, it gives a foundation for the analysis of programs that concurrently access relaxed memory as shown in subsequent work.
\end{abstract}

\section{Introduction}

Concurrency-related bugs are unquestionably one of the most notorious kinds of program errors because most concurrent systems are inherently nondeterministic and their behaviour may therefore not be reproducible. Recent technological advances such as weak memory architectures and highly available distributed services further exacerbate the problem and have renewed interest in the formalization of concurrency semantics.

A recent development in this area includes Concurrent Kleene Algebra (CKA) by Tony Hoare et al.~\cite{HMSW2011}. It is an algebraic semantics of programs that combines the familiar laws of sequential program operators with a new operator for concurrent composition. A distinguishing feature of CKA is its exchange law $(\cU \parallel \cV) ; (\cX \parallel \cY) \subseteq (\cU ; \cX) \parallel (\cV ; \cY)$ that describes how sequential ($;$) and concurrent ($\parallel$) composition operators can be interchanged. Intuitively, the exchange law expresses a divide-and-conquer mechanism for how concurrency may be sequentially implemented on a machine. The exchange law, together with a uniform treatment of programs and their specifications, is key to unifying existing theories of concurrency~\cite{HvS2014}. CKA is such a unifying theory whose universal laws of programming make it well-suited for program correctness proofs. Conversely, however, pure algebra cannot refute that a program is correct or that certain properties about every program always hold~\cite{HvS2012,HvS2014,HvSMSVZOH2014}. This is problematic for theoretical reasons but also in practice because todays software complexity requires a diverse set of program analysis tools that range from proof assistants to automated testing. The solution is to accompany CKA with a mathematical model which satisfies its laws so that we can \emph{prove} as well as \emph{disprove} properties about programs --- the thrust behind this paper.

One well-known model-theoretic foundation for CKA is Gischer's~\cite{G1986} and Pratt's~\cite{P1986} work on modelling concurrency as labelled partially ordered multisets (pomsets). Pomsets generalize the familiar concept of a string in finite automata theory by relaxing the occurrence of alphabet symbols within a string from a total to a partial order. This gives a natural way to not only define sequential but also concurrent composition. The former is a generalization of string concatenation whereas the latter is defined by a form of disjoint union. In addition to their theoretical appeal, partial orders have been shown to be practically useful for distributed systems engineering (e.g.~\cite{F1988,M1989}) and formal software verification (e.g.~\cite{AMSS2010,BOSSW2011}).

This paper therefore adopts Gischer's pomset model to construct a model of CKA that acts as a denotational semantics (due to Scott and Strachey) of concurrent programs. Our construction proceeds in two steps: Section~\ref{section:partial-strings} starts by introducing the general concept of \emph{partial strings} --- similar to partial words~\cite{G1981} and pomsets~\cite{P1986,G1986} --- and Section~\ref{section:programs} then lifts many results of partial strings to downward-closed sets of partial strings, a Hoare powerdomain construction.

One defining characteristic of the partial string model of CKA is particularly worth pointing out. Traditionally, partial words~\cite{G1981} and pomsets~\cite{P1986,G1986} are purely defined in terms of isomorphism classes. In contrast, partial strings are grounded on the concept of \'{E}sik's \emph{monotonic bijective morphisms}~\cite{E2002}. This difference matters for three main reasons: firstly, isomorphisms are about sameness whereas the exchange law on partial strings is an inequation; secondly, our partial string model features least fixed point operators which would render the usual arguments about disjoint representatives of isomorphic classes more subtle because the set of partial strings may be uncountable; lastly, the concept of monotonic bijective morphisms appeals to formalizations with tools that can automatically reason about relations as shown in subsequent work.

We therefore opt for monotonic bijective morphisms as we construct step-by-step a partial string model of CKA. Along the way we leverage the concept of \emph{coproducts} as a means to define partial string operators irrespective of a representative in an isomorphism class, cf.~\cite{G1981,P1986,G1986}. These constructions intentionally shift the emphasis from what we prove about partial strings to \emph{how} we prove these facts. We believe that this can further shape the emerging model-theoretic outlook on CKA and inform its ongoing and future development~\cite{HvSMSVZOH2014}. More concretely, in subsequent work we show how the partial string model of CKA serves as a foundation for the refinement checking of truly concurrent programs with \emph{Satisfiability Modulo Theories} (SMT) solvers.

\paragraph{Related Work} The significance of partial orders for the modelling of concurrency was early on recognized and has extensively flourished ever since in the vast theoretical computer science literature on this topic, e.g.~\cite{P1966,L1978,G1981,NPW1979,G1986,P1986}. The closest work to ours is Gischer's pomset model~\cite{G1986} which strictly generalize Mazurkiewicz traces~\cite{BK2992}. The decidability of inclusion problems for pomset languages with star operators has been most recently established~\cite{LS2014}. More traditionally, recursion and pomsets were treated in the context of ultra-metric spaces~\cite{BW1990}. Winskel's event structures~\cite{W1982} are pomsets enriched with a conflict relation subject to certain conditions. Our partial order abstraction of programs is firmly grounded on \'{E}sik's recent work on infinite partial strings and their monotonic bijective morphisms~\cite{E2002}. The fact that all these works use partial orders to describe the dependency between events means that there is a close connection to the unfolding of petri nets to occurrence nets, an active research area throughout the last four decades, e.g.~\cite{cGP2009}.

\section{Preliminaries}
\label{section:prelim}

Readers who are familiar with lattice and order theory may wish to skip this section. There are comprehensive introductory texts on the subject,~e.g.~\cite{DP2002}.

Denote the set of \defn{natural numbers} by $\nats = \set{1, 2, \ldots}$. The \defn{powerset} of a set $P$ is the set of all subsets of $P$, denoted by $\powerset{P}$. The \defn{empty set} is denoted by $\emptyset$. We write ``$\deq$'' for definitional equality. The \defn{Cartesian product} $X \deq X_1 \times \ldots \times X_n$ of sets $X_1, \ldots, X_n$ is defined to be the set of all ordered $n$-tuples $\tuple{x_1, \ldots, x_n}$ with $x_1 \in X_1, \ldots, x_n \in X_n$. Two elements $\tuple{x_1, \ldots, x_n}$ and $\tuple{y_1, \ldots, y_n}$ of $X$ are defined to be \defn{point-wise equal} whenever the coordinates $x_i$ and $y_i$ are equal for each $1 \leq i \leq n$.

Let $P$ be a set. A \defn{binary relation on $P$} is a subset of $P \times P$. A \defn{preorder} is a binary relation $\preceq$ on $P$ that is \defn{reflexive} ($\forall x \in P \colon x \preceq x$) and \defn{transitive} ($\forall x,y,z \in P \colon (x \preceq y \land y \preceq z) \imp (x \preceq z)$). We write $x \not\preceq y$ when $x \preceq y$ is false. For all $x, y \in P$, $x \prec y$ is called \defn{strict} and is equivalent to $x \preceq y$ and $y \not\preceq x$. A \defn{partial order} is a preorder $\leq$ that is \defn{antisymmetric} ($\forall x,y \in P \colon (x \leq y \land y \leq x) \imp (x = y)$). By reflexivity, partial orders satisfy the converse of the antisymmetry law, i.e. $\forall x,y \in P \colon (x \leq y \land y \leq x) \bimp (x = y)$.

A \defn{partially ordered set}, denoted by $\pair{P}{\leq}$, consists of a set $P$ and a partial order $\leq$. Every logical statement about a partial order $\pair{P}{\leq}$ has a \defn{dual} that is obtained by using $\ge$ instead of $\leq$. A \defn{minimal element} $x \in P$ satisfies $\forall y \in P \colon y \leq x \imp x = y$. Dually, a maximal element $x \in P$ satisfies $\forall y \in P \colon x \leq y \imp x = y$. For all $Q \subseteq P$, $\uparrow_\leq Q \deq \set{y \in P \alt \exists x \in Q \colon x \leq y}$ is the \defn{upward-closed set} of $Q$ in $\pair{P}{\leq}$. As expected, the dual is called the \defn{downward-closed set} of $Q$, denoted by $\downarrow_\leq Q$. Usually, we write $\uparrow Q$ instead of $\uparrow_\leq Q$ when the ordering is clear. Abbreviate $\uparrow x \deq\ \uparrow \set{x}$. For all $H \subseteq P$, $x \in P$ is an \defn{upper bound} of $H$ if $y \leq x$ for all $y \in H$; $x$ is called the \defn{least upper bound} (or \defn{supremum}) of $H$, denoted by $\bigvee H$, if $x$ is an upper bound of $H$ and if, for every upper bound $y$ of $H$, $x \leq y$. By antisymmetry, $\bigvee H$ is unique, if it exists. The \defn{lower bound} and \defn{greatest lower bound}, written as $\bigwedge H$ where $H \subseteq P$, are defined dually. The (unique) \defn{least element} in $P$, if it exists, is $\bot \deq \bigvee \emptyset$ whose dual, if it exists, is $\top \deq \bigwedge \emptyset$. A \defn{lattice} is a partial order $\pair{L}{\leq}$ where every two elements have a (necessarily unique) least upper bound and greatest lower bound: for all $x, y \in L$, these are denoted by $x \vee y$ and $x \wedge y$, respectively. A \defn{complete lattice} $\tuple{L, \leq, \wedge, \vee, \bot, \top}$ is a lattice where $\bigvee S$ and $\bigwedge S$ exists for every $S \subseteq L$.

Let $\pair{P}{\leq}$ and $\pair{Q}{\sqsubseteq}$ be partial orders. A function $f \colon \pair{P}{\leq} \to \pair{Q}{\sqsubseteq}$ is called \defn{monotonic} exactly if $\forall x,y \in P \colon x \leq y \imp f(x) \sqsubseteq f(y)$. Given three sets $P$, $Q$ and $R$, the \defn{composition} of two functions $f \colon P \to Q$ and $g \colon Q \to R$, denoted by $g \circ f$, is a function from $P$ to $R$ such that $g \circ f(x) \deq g(f(x))$ for all $x \in P$.

\section{Partial strings}
\label{section:partial-strings}

We start by abstracting the notion of control flow in concurrent programs as a concept that is similar to partial words~\cite{G1981} and labelled partially ordered multisets (pomsets)~\cite{P1986,G1986}. These concepts generalize the notion of a string by relaxing the total order of alphabet symbols within a string to a partial order. The following definition therefore is fundamental to everything that follows:

\begin{definition}
\label{def:partial-string}
Let $E$ be a nonempty set of \defn{events} and $\Gamma$ be an \defn{alphabet}. Define a \defn{partial string} to be a triple $p = \tuple{E_p, \alpha_p, \preceq_p}$ where $E_p$ is a subset of $E$, $\alpha_p \colon E_p \to \Gamma$ is a function that maps each event in $E_p$ to an alphabet symbol in $\Gamma$, and $\preceq_p$ is a partial order on $E_p$. Two partial strings $p$ and $q$ are said to be \defn{disjoint} whenever $E_p \cap E_q = \emptyset$, $p$ is called \defn{empty} whenever $E_p = \emptyset$, and $p$ is said to be \defn{finite} whenever $E_p$ is finite. Let $\sP$ be the set of all partial strings, and denote with $\sP_f$ the set of all finite partial strings in $\sP$.
\end{definition}

\begin{SCfigure}[100][b]
\xymatrix@R=1.2em@C=1em{
                  &            e_0              &                  \\
  e_1\ar@{<-}[ur] &                             &  e_2\ar@{<-}[ul] \\
                  & e_3\ar@{<-}[lu]\ar@{<-}[ru] &
}
\caption{
Upside down Hasse diagram of a partial string $\tuple{E_p, \alpha_p, \preceq_p}$. Assume alphabet $\Gamma = \set{\texttt{read}, \texttt{write}} \times \set{\texttt{x}, \texttt{y}}$ is the Cartesian product of labels that distinguish reads from writes on two different shared memory locations \texttt{x} and \texttt{y}. The ordering of events includes $e_0 \preceq_p e_3$, but $e_1$ and $e_2$ are incomparable.
}
\label{fig:partial-string-example}
\end{SCfigure}
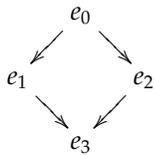

Each event in the universe $E$ should be thought of as an occurrence of a computational step, whereas the alphabet $\Gamma$ could be seen as a way to label events. Typically we denote partial strings in $\sP$ (whether finite or not) by $p$ or $q$, or letters from $u$ through $z$. In essence, a partial string $p$ is a partially ordered set $\pair{E_p}{\preceq_p}$ equipped with a function $\alpha_p$ that maps every event in $E_p$ to an alphabet symbol in $\Gamma$. For reasons that become clearer in subsequent work, it is convenient to draw finite partial strings as upside down Hasse diagrams (e.g. Figure~\ref{fig:partial-string-example}), where the ordering between events should be interpreted as a causality relation such as the \emph{sequenced-before} relation in C++11~\cite{BOSSW2011}. For example, $e_0 \preceq_p e_3$ in partial string $p$ means that $e_0$ (top) is sequenced-before $e_3$ (bottom), whereas $e_1$ and $e_2$ are unsequenced (i.e. happen concurrently) because neither $e_1 \preceq_p e_2$ nor $e_2 \preceq_p e_3$. The alphabet $\Gamma$, in turn, gives a secondary-level interpretation of events. For example, the alphabet in Figure~\ref{fig:partial-string-example} describes the computational effects of events in terms of shared memory accesses. More abstractly, if we see the alphabet $\Gamma$ as a set of labels, then $\alpha_p$ for a partial string $p$ is like the labelling function of partial words~\cite{G1981} and pomsets~\cite{P1986,G1986}.

Unlike partial words and pomsets, however, partial strings retain the identity of events which means that operators on partial strings can be defined irrespective of a representative in an isomorphism class, cf.~\cite{G1981,P1986,G1986}. Therefore two partial strings are equal whenever they are point-wise equal. Of course, point-wise equality is too coarse for many practical purposes and so we later introduce the concept of a monotonic bijective morphism (Definition~\ref{def:partial-string-isomorphism}). Until we do so, however, we can still make two useful observations.

\begin{proposition}
\label{proposition:unique-empty-partial-string}
The empty partial string, denoted by $\bot$, is unique.
\end{proposition}
\begin{proof}
The existence of $\bot$ is entirely trivial. Assume $\bot'$ is another empty string. By Definition~\ref{def:partial-string}, $E_\bot = E_{\bot'} = \emptyset$. Therefore, the partial orders of both $\bot$ and $\bot'$ are empty, i.e. $\preceq_\bot = \preceq_{\bot'} = \emptyset$. And trivially $\alpha_\bot(e) = \alpha_{\bot'}(e)$ for all $e \in E$. By coordinate-wise equality, $\bot = \bot'$.
\end{proof}

It is good to be aware of the cardinality of the set of partial strings.

\begin{proposition}
\label{proposition:cardinality}
If the alphabet $\Gamma$ and universe of events $E$ is countably infinite, then $\sP$ has the cardinality of the continuum.
\end{proposition}
\begin{proof}
Assume $\Gamma$ and $E$ are countably infinite. The set of partial orders on $E$ is a subset of $\powerset{E \times E}$, whereas the set of labelling functions is $E \to \Gamma$. By assumption, $\powerset{E \times E}$ has the same cardinality as $\powerset{E}$, and $E \to \Gamma$ has the same cardinality as $\powerset{E \times \Gamma}$ because functions are merely special relations. By assumption, $\powerset{E \times \Gamma}$ has the same cardinality as $\powerset{E}$. The fact that the cardinality of the powerset of a countably infinite set has the same cardinality as the continuum goes back to the Cantor-Schr\"{o}der-Bernstein Theorem.
\end{proof}

\begin{remark}
Since the concept of a partial string is closely related to that of a partial order, it is interesting that Roscoe notes that the set of strict partial orders over a fixed universe forms a dcpo~\cite[p. 474]{RHB1997}.\footnote{Recall that a nonempty subset $D$ of a partially ordered set is called \defn{directed} if each finite subset $F$ of $D$ has an upper bound in $D$, i.e. $\exists y \in D \colon \forall x \in F \colon x \leq y$. For example, the set of all finite subsets of natural numbers, ordered by subset inclusion, is directed. A \defn{directed-complete partial order} (often abbreviated \defn{dcpo}) is a partial order with a bottom element and in which every directed set has a least upper bound. An example of a dcpo is the set of all partial functions ordered by subset inclusion of their respective graphs~\cite[pp. 180f]{DP2002}.} He also mentions that it is far from trivial to show that the maximal elements in this dcpo are the total orders on the universe, a fact that requires the Axiom of Choice.
\end{remark}

The purpose of Proposition~\ref{proposition:cardinality} is to caution us concerning the treatment of partial strings in the infinite case. We will see shortly how this precaution plays out in the definitions and proofs about partial strings where we purposefully avoid relying on a representative in an isomorphism class.

Given a partial string $x$, recall that its set of events is denoted by $E_x$ where $E_x \subseteq E$ because the events in every partial string are always drawn from the universe of events $E$. Similarly, the other two tuple components in a partial string $x$ are identified with a subscript, i.e. $\alpha_x$ is $x$'s labelling function whereas $\preceq_x$ is a partial order on $E_x$. We are about to use these three tuple components to define a binary relation on the set of partial strings, $\sP$, that is very similar to ``subsumption'' in the equational theory of pomsets~\cite{G1986} except that ours disregards the identity of events. Informed by our precaution mentioned above, this is formalized by monotonic bijective morphisms~\cite{E2002} that generalize the concept of isomorphisms as follows:

\begin{definition}
\label{def:partial-string-isomorphism}
Let $x,y \in \sP$ be partial strings such that $x = \tuple{E_x, \alpha_x \preceq_x}$ and $y = \tuple{E_y, \alpha_y, \preceq_y}$. Then $x$ and $y$ are \defn{isomorphic}, denoted by $x \cong y$, if there exists an order-isomorphism between $x$ and $y$ that preserves their labeling, i.e. there exists a one-to-one and onto function (\defn{bijection}) $f \colon E_x \to E_y$ such that, for all $e, e' \in E_x$, $e \preceq_x e' \bimp f(e) \preceq_y f(e')$ and $\alpha_x(e) = \alpha_y(f(e))$. Define a \defn{monotonic bijective morphism}, written $f \colon x \to y$, to be a bijection $f$ from $E_x$ to $E_y$ such that, for all $e, e' \in E_x$, $e \preceq_x e' \imp f(e) \preceq_y f(e')$ and $\alpha_x(e) = \alpha_y(f(e))$. We write $x \sqsubseteq y$ whenever there exists a monotonic bijective morphism $f \colon y \to x$ from $y$ to $x$.
\end{definition}

\begin{SCfigure}[100][t]
\xymatrix@R=1.2em@C=1em{
  e_0            & e_1            \\
  e_2\ar@{<-}[u] & e_3\ar@{<-}[u]
}
\caption{
A partial string $p = \tuple{E_p, \alpha_p, \preceq_p}$ with $E_p \deq \set{e_0, e_1, e_2, e_3}$. For alphabet $\Gamma = \set{a,b}$, assume the labeling function is defined by $\alpha_p(e_0) = \alpha_p(e_1) = a$ and $\alpha_p(e_2) = \alpha_p(e_3) = b$.
}
\label{fig:almost-N-partial-string}
\end{SCfigure}

\begin{SCfigure}[100][b]
\xymatrix@R=1.2em@C=1em{
  e_0\ar@{->}[dr]& e_1           \\
  e_2\ar@{<-}[u] & e_3\ar@{<-}[u]
}
\caption{
Define $N(a, a, b, b) \deq \tuple{E_p, \alpha_p, \preceq_p}$ with $E_p = \set{e_0, e_1, e_2, e_3}$ and $\preceq_p\ \deq \set{\pair{e_0}{e_0}, \pair{e_1}{e_1}, \pair{e_2}{e_2}, \pair{e_3}{e_3}, \pair{e_0}{e_2}, \pair{e_0}{e_3}, \pair{e_1}{e_3}}$ such that $\alpha_p(e_0) = \alpha_p(e_1) = a$ and $\alpha_p(e_2) = \alpha_p(e_3) = b$ where $a, b \in \Gamma$.
}
\label{fig:N-partial-string}
\end{SCfigure}

Intuitively, $\sqsubseteq$ orders partial strings according to the sequenced-before relation between events. In other words, $x \sqsubseteq y$ for partial strings $x$ and $y$ could be interpreted as saying that all events ordered in $y$ have the same order in $x$. This way $x \sqsubseteq y$ acts as a form of refinement ordering between partial strings where $x$ refines $y$, or $x$ is more deterministic than $y$.

\begin{remark}
For arbitrary partial strings, the refinement ordering $\sqsubseteq$ from Definition~\ref{def:partial-string-isomorphism} is different from the containment of two pomset languages, cf.~\cite{FKL1993}. For example, the languages of the partial strings $(a ; b) \parallel (a ; b)$ and $N(a, a, b, b)$ --- as shown in Figure~\ref{fig:almost-N-partial-string}~and~\ref{fig:N-partial-string}, respectively --- are contained in each other but $(a ; b) \parallel (a ; b) \not\sqsubseteq N(a, a, b, b)$, see also Pratt~\cite[p. 13]{P1986}. More accurately, Definition~\ref{def:partial-string-isomorphism} implies pomset language containment, but not vice versa.
\end{remark}

The following important fact about $\sqsubseteq$ is clear since every finite number of function compositions preserves monotonicity and bijectivity.

\begin{proposition}
\label{proposition:sqsubseteq-preorder}
$\pair{\sP}{\sqsubseteq}$ is a preorder.
\end{proposition}
\begin{proof}
Let $x,y,z \in \sP$ be partial strings. Since the identity function is bijective, monotonic and label-preserving, $x \sqsubseteq x$, proving reflexivity. Assume $x \sqsubseteq y$ and $y \sqsubseteq z$. By Definition~\ref{def:partial-string-isomorphism}, there exist two monotonic bijective morphisms $f \colon y \to x$ and $g \colon z \to y$. By function composition, $f \circ g \colon z \to x$ is a monotonic bijective morphism, proving transitivity.
\end{proof}

Even though $\sqsubseteq$ is a preorder, it is generally \emph{not} a partial order unless we impose further restrictions on partial strings~\cite{E2002}. In particular, the following proposition allows us to treat all \emph{finite} partial strings as a partially ordered set because the refinement order $\sqsubseteq$ on $\sP_f$ is antisymmetric, cf. Proposition~3.5's proof~in~\cite{E2002}:

\begin{proposition}[Transitivity]
\label{proposition:sqsubseteq-partial-order}
For all $x, y \in \sP_f$, if $x \sqsubseteq y$ and $y \sqsubseteq x$, then $x \cong y$.
\end{proposition}
\begin{proof}
Let $x, y \in \sP_f$ be finite partial strings. Let $f \colon x \to y$ and $g \colon y \to x$ be monotonic bijective morphisms as witnesses for $y \sqsubseteq x$ and $x \sqsubseteq y$, respectively. Let $e, e' \in E_x$ be events. By Definition~\ref{def:partial-string-isomorphism}, it suffices to show $f(e) \precsim_y f(e')$ implies $e \precsim_x e'$. Assume $f(e) \precsim_y f(e')$. Let $h \deq g \circ f$. Define $h^1 \deq h$ and $h^{n+1} \deq h^n \circ h$ for all $n \in \nats$. Since $E_x$ and $E_y$ are finite, there exists $n$ such that $h^n$ is the identity function on $E_x$. Fix $n$ to be the smallest such natural number. Assume $n = 1$. Then $g(f(e)) \precsim_x g(f(e'))$ by assumption and $g$'s monotonicity. Since $h$ is the identity function, it follows $e \precsim_x e'$, as required. Now assume $n > 1$. Since every finite number of function compositions preserve monotonicity, $h^{n-1}(g(f(e))) \precsim_x h^{n-1}(g(f(e)))$. Since $h^n$ is the identity function, $e \precsim_x e'$, proving that $f$ is an isomorphism. We conclude that $x \cong y$.
\end{proof}

Figure~\ref{fig:sqsubseteq-partial-string} illustrates the essence of Proposition~\ref{proposition:sqsubseteq-partial-order}. And clearly its converse also holds (even for infinite partial strings) because a bijective monotonic morphism generalizes the concept of an isomorphism.

\begin{proposition}
\label{proposition:converse-sqsubseteq-partial-order}
For all $x, y \in \sP$, if $x \cong y$, then $x \sqsubseteq y$ and $y \sqsubseteq x$.
\end{proposition}
\begin{proof}
Let $x$ and $y$ be partial strings. Assume $x \cong y$. Fix $f$ to be an isomorphism from $x$ to $y$. By Definition~\ref{def:partial-string-isomorphism}, $f$ is a bijective monotonic morphism, and so is its inverse $f^{-1} \colon y \to x$, proving $y \sqsubseteq x$ and $x \sqsubseteq y$, respectively.
\end{proof}

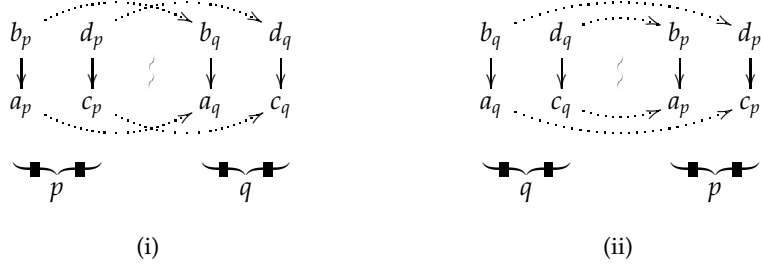
\begin{figure}
\centering
\subcaptionbox{\label{fig:sqsubseteq-partial-order-1}}{
\xy
  \xymatrix "M"@C=1.2em@R=1em{
    b_p\ar@{.>}@/^1pc/[rrr]             & d_p\ar@{.>}@/^1pc/[rrr]             &              & b_q            & d_q           \\
    a_p\ar@{<-}[u]\ar@{.>}@/^-1pc/[rrr] & c_p\ar@{<-}[u]\ar@{.>}@/^-1pc/[rrr] &   \ar@{~}[u] & a_q\ar@{<-}[u] & c_q\ar@{<-}[u] \\
&&&&\\
  }
\POS"M3,1"."M3,2"!C*\frm{_\}},+D*++!U\txt{$p$}
    ,"M3,4"."M3,5"!C*\frm{_\}},+D*++!U\txt{$q$}
\endxy
}
\qquad
\qquad
\qquad
\subcaptionbox{\label{fig:sqsubseteq-partial-order-2}}{
\xy
  \xymatrix "M"@C=1.2em@R=1em{
    b_q\ar@{.>}@/^1pc/[rrrr]             & d_q\ar@{.>}@/^0.5pc/[rr]             &              & b_p            & d_p           \\
    a_q\ar@{<-}[u]\ar@{.>}@/^-1pc/[rrrr] & c_q\ar@{<-}[u]\ar@{.>}@/^-0.5pc/[rr] &   \ar@{~}[u] & a_p\ar@{<-}[u] & c_p\ar@{<-}[u] \\
&&&&\\
  }
\POS"M3,1"."M3,2"!C*\frm{_\}},+D*++!U\txt{$q$}
    ,"M3,4"."M3,5"!C*\frm{_\}},+D*++!U\txt{$p$}
\endxy
}
\caption{
The left and right show a monotonic bijective morphism $f \colon p \to q$ and $g \colon q \to p$, respectively. Thus $q \sqsubseteq p$ and $p \sqsubseteq q$. Let $h \deq g \circ f$. Then $h^2 = h \circ h$ is the identity function on $E_p$. Both $f$ and $g$ are isomorphisms, whence $p \cong q$.
}
\label{fig:sqsubseteq-partial-string}
\end{figure}

We purposefully did not define $\sP$ as a family of disjoint partial strings because the resulting structure with the point-wise join operator would not generally preserve disjointness --- a problem since the union of overlapping partial orders is generally not antisymmetric. Moreover, we would like to avoid having to choose disjoint representatives of isomorphism classes from a (possibly uncountable) set of partial strings. Yet we would like somehow to be able to compose partial strings to form new ones. For this purpose, we define compositions as coproducts, thereby guaranteeing disjointness by construction.

\begin{definition}
\label{def:partial-string-composition}
Define $\bbB \deq \set{0,1}$ to be integers zero and one. For every set $S$ and $T$, define the \defn{coproduct} of $S$ and $T$, written $S + T$, to be the set union of $S \times \set{0}$ and $T \times \set{1}$. Given two partial strings $x$ and $y$ in $\sP$, define their \defn{concurrent} and \defn{strongly sequential composition} by $x \parallel y \deq \tuple{E_{x \parallel y}, \alpha_{x \parallel y}, \precsim_{x \parallel y}}$ and $x ; y \deq \tuple{E_{x ; y}, \alpha_{x ; y}, \precsim_{x ; y}}$, respectively, where $E_{x \parallel y} = E_{x ; y} \deq E_x + E_y$ are coproducts such that, for all $e, e' \in E_x \cup E_y$ and $i, j \in \bbB$, the following holds:
\begin{itemize}
\item $\pair{e}{i} \preceq_{x \parallel y} \pair{e'}{j}$ exactly if ($i = j = 0$ and $e \preceq_x e'$) or ($i = j = 1$ and $e \preceq_y e'$),
\item $\pair{e}{i} \preceq_{x ; y} \pair{e'}{j}$ exactly if $i < j$ or $\pair{e}{i} \preceq_{x \parallel y} \pair{e'}{j}$,
\item $\alpha_{x \parallel y}(\pair{e}{i}) = \alpha_{x ; y}(\pair{e}{i}) \deq
\begin{cases}
  \alpha_x(e) &\text{if } i = 0 \\
  \alpha_y(e) &\text{if } i = 1.
\end{cases}$
\end{itemize}
\end{definition}

For the concept of coproducts to be applicable here, we require the set of events $E$ to be infinite, cf.~Proposition~\ref{proposition:cardinality}. Given two sets $S$ and $T$, their coproduct $S + T \deq (S \times \set{0}) \cup (T \times \set{1})$ is like a `constructive disjoint union' in the sense that it explicitly identifies the elements in $S$ and $T$ through the integers $0$ and $1$ in $\bbB$, respectively. This significantly shapes the sort of proofs we get about sequential ($;$) and concurrent ($\parallel$) composition throughout the rest of this section as we shall start to see soon. We first establish that the set of partial strings is closed under both newly defined binary operators.

\begin{proposition}[$\sP$ is closed under partial string operators]
For every partial string $x$ and $y$ in $\sP$, $x \parallel y$ and $x ; y$ are also in $\sP$.
\end{proposition}
\begin{proof}
Show that $\precsim_{x \parallel y}$ and $\precsim_{x ; y}$ are partial orders. By assumption, $\precsim_x$ and $\precsim_y$ are partial orders. By case analysis, $\precsim_{x \parallel y}$ is reflexive, transitive and antisymmetric. Since $\precsim_{x \parallel y}$ is a subset of $\precsim_{x ; y}$, it follows that $\precsim_{x ; y}$ is reflexive and transitive. Let $e,e' \in E_x \cup E_y$ and $i,i' \in \bbB$. Assume $\pair{e}{i} \precsim_{x ; y} \pair{e'}{i'}$ and $\pair{e'}{i'} \precsim_{x ; y} \pair{e}{i}$. Definition~\ref{def:partial-string-composition} implies $i = i'$. From antisymmetry of $\precsim_x$ and $\precsim_y$ follows $e = e'$. By point-wise equality, $\pair{e}{i} = \pair{e'}{i'}$. We conclude that $x \parallel y$ and $x ; y$ are partial strings in $\sP$ according to Definition~\ref{def:partial-string}.
\end{proof}

We can therefore speak of $;$ and $\parallel$ as \emph{partial string operators}. Before we begin our study of these operators, it helps to develop some intuition for them. Figure~\ref{fig:partial-string-composition} illustrates the sequential and concurrent composition of two simple partial strings as coproducts. Note that $p ; q \sqsubseteq p \parallel q$ in Figure~\ref{fig:partial-string-composition}. Semantically, the operators $\parallel$ and $;$ are identical to concurrent and sequential composition on pomsets, respectively, as defined by Gisher~\cite{G1986}, except of the aforementioned role of coproducts. When clear from the context, we construct partial strings directly from labels. For example, as already seen $(a ; b) \parallel (a ; b)$ for $a, b \in \Gamma$ corresponds to a partial string that is isomorphic to the one shown in Figure~\ref{fig:almost-N-partial-string}.

The significance of the three basic definitions given so far is best illustrated through proofs about properties of $\parallel$ and $;$. So we start with a simple proof of the fact that $x \parallel y$ and $y \parallel x$ are isomorphic for every partial string $x$ and $y$ (whether finite or not).

\begin{figure}
\centering
\subcaptionbox{\label{fig:partial-string-composition-1}}{
\xymatrix@R=1.2em@C=1em{
  e_0                                     \\
  e_1\ar@{<-}[u]
}
}
\qquad
\qquad
\subcaptionbox{\label{fig:partial-string-composition-2}}{
\xymatrix@R=1.2em@C=1em{
  e_2                                     \\
}
}
\qquad
\qquad
\subcaptionbox{\label{fig:partial-string-composition-3}}{
\xymatrix@R=1.2em@C=1em{
  \pair{e_0}{0}           &                 \\
  \pair{e_1}{0}\ar@{<-}[u] & \pair{e_2}{1}  \\
}
}
\qquad
\qquad
\subcaptionbox{\label{fig:partial-string-composition-4}}{
\xymatrix@R=1.2em@C=1em{
  \pair{e_0}{0}\ar@{->}[d]                \\
  \pair{e_1}{0}\ar@{->}[d]                \\
  \pair{e_2}{1}                           \\
}
}
\caption{
Let $p$ and $q$ be finite partial strings as shown in~\ref{fig:partial-string-composition-1}~and~\ref{fig:partial-string-composition-2}, respectively. Then~\ref{fig:partial-string-composition-3}~and~\ref{fig:partial-string-composition-2} illustrate the coproduct $p \parallel q$ and $p ; q$, respectively.
}
\label{fig:partial-string-composition}
\end{figure}

\begin{proposition}[$\parallel$-commutativity]
\label{proposition:partial-string-concurrent-commutativity}
For all $x, y \in \sP$, $x \parallel y \cong y \parallel x$.
\end{proposition}
\begin{proof}
To show that concurrent composition of partial strings is commutative, let $f \colon E_{x \parallel y} \to E_{y \parallel x}$ be a function such that, for all $e \in E_x \cup E_y$ and $i \in \bbB$, $f(\pair{e}{i}) = \pair{e}{1 - i}$. Clearly $f$ is bijective. It remains to show that $f$ is a label-preserving order-isomorphism. Let $e' \in E_x \cup E_y$ and $i' \in \bbB$. Then
\begin{align}
\pair{e}{i} \precsim_{x \parallel y} \pair{e'}{i'}
&\bimp \pair{e}{1 - i} \precsim_{y \parallel x} \pair{e'}{1 - i'} \ptag{Definition~\ref{def:partial-string-composition} of $\parallel$ with $i = i'$} \\
&\bimp f(\pair{e}{i}) \precsim_{y \parallel x} f(\pair{e'}{i'}). \ptag{Definition of $f$}
\end{align}
Moreover $\alpha_{x \parallel y}(\pair{e}{i}) = \alpha_{y \parallel x}(\pair{e}{1 - i}) = \alpha_{y \parallel x}(f(\pair{e}{i}))$. Hence $x \parallel y \cong y \parallel x$.
\end{proof}

We have given the details of the proof to draw attention to the fact that the label-preserving order-isomorphism acts as a witness for the truth of the statement. We call this a \emph{constructive proof}. Along similar lines, it is not difficult to constructively prove (Proposition~\ref{proposition:partial-string-identity}) that $\bot$ is the identity element (up to isomorphism) for both the sequential and concurrent partial string composition operators. In fact, since it is a recurring theme that an algebraic property holds for both operators, it is convenient to define the following:

\begin{definition}[Bow tie]
\label{def:bowtie}
For all partial strings $x$ and $y$, denote with $x \Join y$ either concurrent or sequential composition of $x$ and $y$. That is, a statement about $\Join$ is the same as two statements where $\Join$ is replaced by $\parallel$ and $;$, respectively.
\end{definition}

Here is a proof of the identity element so that the reader may get familiar with our use of the ``bow tie'' placeholder.

\begin{proposition}[Identity]
\label{proposition:partial-string-identity}
For all $x \in \sP$, $x \Join \bot \cong \bot \Join x \cong x$.
\end{proposition}
\begin{proof}
Let $f \colon E_x \to E_{x \Join \bot}$ be a function such that, for all $e \in E_x$, $f(e) = \pair{e}{0}$. By Definition~\ref{def:partial-string-composition}, $E_{x \Join \bot} = E_x \times \set{0}$ because $E_\bot \times \set{1} = \emptyset$. Clearly $f$ is a label-preserving order-isomorphism, whence $x \Join \bot \cong x$. Similarly, $\bot \Join x \cong x$.
\end{proof}

Equivalently, by Definition~\ref{def:bowtie}, $x ; \bot \cong \bot ; x \cong x$ and $x \parallel \bot \cong \bot \parallel x \cong x$ for all $x \in \sP$. Moreover, the two binary operators on partial strings are related in the expected way as already witnessed in the example of Figure~\ref{fig:partial-string-composition}:

\begin{proposition}[Basic refinement]
\label{proposition:partial-string-operator-chain}
For all $x, y \in \sP$, $x ; y \sqsubseteq x \parallel y$.
\end{proposition}
\begin{proof}
Let $e, e' \in E_x \cup E_y$, $i, i' \in \bbB$, and $f \colon E_{x \parallel y} \to E_{x ; y}$ is a function such that $f(\pair{e}{i}) = \pair{e}{i}$. Clearly $f$ is a bijection. By Definition~\ref{def:partial-string-composition}, $f$ preserves labels because $\alpha_{x \parallel y}(\pair{e}{i}) = \alpha_{x ; y}(\pair{e}{i}) = \alpha_{x ; y}(f(\pair{e}{i}))$, and $\pair{e}{i} \precsim_{x \parallel y} \pair{e'}{i'}$ implies $f(\pair{e}{i}) \precsim_{x ; y} f(\pair{e'}{i'})$, proving its monotonicity. Therefore, by Definition~\ref{def:partial-string-isomorphism}, $f \colon x \parallel y \to x ; y$ is a monotonic bijective morphism, proving the claim.
\end{proof}

We have given the detailed proof of the previous proposition because it is a template for more complicated ones. In particular, it illustrates how the concept of a monotonic bijective morphism lends itself for constructive proofs about \emph{inequalities} between partial strings. This way we can carry forward the simplicity and ingenuity of the pomset model to inequational reasoning. However, this added flexibility also means that proofs about partial strings may end up sometimes more combinatorial, as illustrated next.

\begin{proposition}[Monotonicity]
\label{proposition:partial-string-monotonicity}
For all $x, y, z \in \sP$, if $x \sqsubseteq y$, then $x \Join z \sqsubseteq y \Join z$ and $z \Join x \sqsubseteq z \Join y$.
\end{proposition}
\begin{proof}
Assume $x \sqsubseteq y$. Show $x \parallel z \sqsubseteq y \parallel z$. Let $g \colon y \to x$ be a monotonic bijective morphism as a witness for the assumption. Let $f \colon E_{y \parallel z} \to E_{x \parallel z}$ be a function such that, for all $e \in E_y \cup E_z$ and $i \in \bbB$, $$f(\pair{e}{i}) =
\begin{cases}
  \pair{g(e)}{0} &\text{if } i = 0 \\
  \pair{e}{1} &\text{if } i = 1
\end{cases}.$$ Clearly $f$ is bijective and it preserves labels. Let $e' \in E_y \cup E_z$ and $i' \in \bbB$. Assume $\pair{e}{i} \precsim_{y \parallel z} \pair{e}{i'}$. By Definition~\ref{def:partial-string-composition} and the last assumption, there are only two cases to consider: either $i = i' = 0$ or $i = i' = 1$. If $i = i' = 0$, then $e \precsim_y e'$ and $\pair{g(e)}{0} = f(\pair{e}{0}) \precsim_{x \parallel z} f(\pair{e'}{0}) = \pair{g(e')}{0}$ because $g(e) \precsim_x g(e')$; otherwise, $\pair{e}{1} = f(\pair{e}{1}) \precsim_{x \parallel z} f(\pair{e'}{1}) = \pair{e'}{1}$ because $e \precsim_z e'$, proving that $x \sqsubseteq y$ implies $x \parallel z \sqsubseteq y \parallel z$. By $\parallel$-commutativity, if $x \sqsubseteq y$, then $z \parallel x \sqsubseteq z \parallel y$. The proof for $;$-monotonicity is similar except that there are three cases to consider.
\end{proof}

In addition to being monotonic, both sequential and concurrent composition are associative up to isomorphism.

\begin{proposition}[Associativity]
\label{proposition:partial-string-associativity}
For all $x, y, z \in \sP$, $(x \Join y) \Join z \cong x \Join (y \Join z)$.
\end{proposition}
\begin{proof}
Let $f \colon E_{(x \parallel y) \parallel z} \to E_{x \parallel (y \parallel z)}$ be a function such that, for all $e \in E_{(x \parallel y) \parallel z}$, $$f(e) =
\begin{cases}
  \pair{\pair{e'}{0}}{0} &\text{if } \exists e' \in E_x \colon \pair{e'}{0} = e \\
  \pair{\pair{e'}{1}}{0} &\text{if } \exists e' \in E_y \colon \pair{\pair{e'}{0}}{1} = e \\
  \pair{e'}{1} &\text{if } \exists e' \in E_z \colon \pair{\pair{e'}{1}}{1} = e
\end{cases}.$$ We leave it as an exercise to the reader to verify that $f$ is a partial string isomorphism, proving $\parallel$-associativity. Similarly for $;$-associativity.
\end{proof}

\begin{definition}
\label{def:weakly-sequential-composition}
For all $x, y \in \sP$, \defn{weakly sequential composition}, written $x \fatsemi y$, is any coproduct according to Definition~\ref{def:partial-string-composition} that satisfies $x ; y \sqsubseteq x \fatsemi y \sqsubseteq x \parallel y$.
\end{definition}

Usually, the weakly sequential composition of partial strings is expected to be strictly more deterministic than their concurrent composition. Of course, it need not always be the case that $x \parallel y$, $x \fatsemi y$ and $x ; y$ are non-isomorphic, particularly since $\bot$ is the left and right identity element of both strongly and weakly sequential compositions of finite partial strings.

\begin{corollary}[Weakly sequential identity]
\label{corollary:partial-string-identity}
For all $x \in \sP_f$, $x \fatsemi \bot \cong \bot \fatsemi x \cong x$.
\end{corollary}
\begin{proof}
By Definition~\ref{def:weakly-sequential-composition} and Proposition~\ref{proposition:partial-string-identity}, $x \cong x ; \bot \sqsubseteq x \fatsemi \bot \sqsubseteq x \parallel \bot \cong x$. By Proposition~\ref{proposition:sqsubseteq-partial-order}, $x \fatsemi \bot \cong x$. Similarly, $\bot \fatsemi x \cong x$.
\end{proof}

The concept of weakly sequential composition is relevant for the discussion of fences which are, however, beyond the scope of this paper. Instead we focus on strongly sequential composition and especially its interplay with concurrent composition. For this, we prove that concurrent and sequential composition of partial strings abides to CKA's exchange law which is modeled after the interchange law in two-category theory~\cite{HMSW2011}. For the theory of concurrency, the exchange law is important because it shows how concurrent and sequential composition can be interchanged. Operationally, it could be seen as a divide-and-conquer mechanism for how concurrent composition may be sequentially implemented on a machine.

\begin{proposition}
\label{proposition:exchange-law}
For all $u, v, x, y \in \sP$, $(u \parallel v) ; (x \parallel y) \sqsubseteq (u ; x) \parallel (v ; y)$.
\end{proposition}
\begin{proof}
Let $f \colon E_{(u ; x) \parallel (v ; y)} \to E_{(u \parallel v) ; (x \parallel y)}$ such that $f(\pair{\pair{e}{i}}{j}) \deq \pair{\pair{e}{j}}{i}$ for all events $e \in E_u \cup E_v \cup E_x \cup E_y$ and $i, j \in \bbB$. Clearly $f$ is bijective. Moreover, the following equalities hold by Definition~\ref{def:partial-string-composition} of the labelling function:
\begin{align}
\alpha_{(u ; x) \parallel (v ; y)}(\pair{\pair{e}{i}}{j}) &=
\begin{cases}
  \alpha_{u ; x}(\pair{e}{i}) &\text{if } j = 0 \ptag{Definition of $\parallel$} \\
  \alpha_{v ; y}(\pair{e}{i}) &\text{if } j = 1
\end{cases} \\
&=\begin{cases}
  \alpha_u(e) &\text{if } i = 0 \text{ and } j = 0 \\
  \alpha_v(e) &\text{if } i = 0 \text{ and } j = 1 \\
  \alpha_x(e) &\text{if } i = 1 \text{ and } j = 0 \\
  \alpha_y(e) &\text{if } i = 1 \text{ and } j = 1
\end{cases} \ptag{Definition of $;$} \\
&=\begin{cases}
  \alpha_{u \parallel v}(\pair{e}{j}) &\text{if } i = 0 \\
  \alpha_{x \parallel y}(\pair{e}{j}) &\text{if } i = 1
\end{cases} \ptag{Definition of $\parallel$} \\
&=\alpha_{(u \parallel v) ; (x \parallel y)}(\pair{\pair{e}{j}}{i}) \ptag{Definition of $;$} \\
&=\alpha_{(u \parallel v) ; (x \parallel y)}(f(\pair{\pair{e}{i}}{j})) \ptag{Definition of $f$}.
\end{align}
In short, $f$ preserves the labelling of events. Let $e' \in E_u \cup E_v \cup E_x \cup E_y$ and $i', j' \in \bbB$. Assume $\pair{\pair{e}{i}}{j} \precsim_{(u ; x) \parallel (v ; y)} \pair{\pair{e'}{i'}}{j'}$. We must show $f(\pair{\pair{e}{i}}{j}) \precsim_{(u \parallel v) ; (x \parallel y)} f(\pair{\pair{e'}{i'}}{j'})$. By definition of $f$, it suffices to show $\pair{\pair{e}{j}}{i}) \precsim_{(u \parallel v) ; (x \parallel y)} \pair{\pair{e'}{j'}}{i'}$. By assumption and Definition~\ref{def:partial-string-composition} of concurrent composition, ($j = j' = 0$ and $\pair{e}{i} \precsim_{(u ; x)} \pair{e'}{i'}$) or ($j = j' = 1$ and $\pair{e}{i} \precsim_{(v ; y)} \pair{e'}{i'}$). By Definition~\ref{def:partial-string-composition} of strongly sequential composition, it follows
\begin{align*}
\bigg(j = j' = 0 \text{ and } \big(i < i' \text{ or } (i = i' = 0 \text{ and } e \precsim_u e') \text{ or } (i = i' = 1 \text{ and } e \precsim_x e')\big)\bigg) \text{ or } \\
\bigg(j = j' = 1 \text{ and } \big(i < i' \text{ or } (i = i' = 0 \text{ and } e \precsim_v e') \text{ or } (i = i' = 1 \text{ and } e \precsim_y e')\big)\bigg).
\end{align*}
From propositional logic follows
\begin{align*}
i < i' \text{ or } \bigg(i = i' = 0 \text{ and } \big((j = j' = 0 \text{ and } e \precsim_u e') \text{ or } (j = j' = 1 \text{ and } e \precsim_v e')\big)\bigg) \text{ or } \\
\bigg(i = i' = 1 \text{ and } \big((j = j' = 0 \text{ and } e \precsim_x e') \text{ or } (j = j' = 1 \text{ and } e \precsim_y e')\big)\bigg).
\end{align*}
By Definition~\ref{def:partial-string-composition} of concurrent composition, $i < i'$ or $\pair{e}{j} \precsim_{u \parallel v} \pair{e'}{j'}$ or $\pair{e}{j} \precsim_{x \parallel y} \pair{e'}{j'}$. Thus, by definition of strongly sequential composition, $\pair{\pair{e}{j}}{i} \precsim_{(u \parallel v) ; (x \parallel y)} \pair{\pair{e'}{j'}}{i'}$, whence $f(\pair{\pair{e}{i}}{j}) \precsim_{(u \parallel v) ; (x \parallel y)} f(\pair{\pair{e'}{i'}}{j'})$ by definition of $f$. Therefore $f \colon (u ; x) \parallel (v ; y) \to (u \parallel v) ; (x \parallel y)$ is a monotonic bijective morphism, proving the claim by Definition~\ref{def:partial-string-isomorphism}.
\end{proof}

The exchange law is originally postulated by Gisher in his thesis where it is called ``subsumption axiom''~\cite[p. 22]{G1985}. Here we prove that his exchange law holds for partial strings with respect to the refinement order $\sqsubseteq$. It is also interesting to compare the previous constructive proof with \'{E}sik's argument why the exchange law is morally true~\cite[Proposition~3.7]{E2002}. In the context of formal verification tools, however, we prefer the previous construction because it retains more of the computational aspect of the problem.

Last but not least, it is not difficult to convince ourselves that the exchange law for partial strings is not an order-isomorphism. In fact, if it were, the Eckmann-Hilton argument about any two monoid structures would imply that sequential and concurrent composition coincide --- something we usually would not want because it would conflate concepts that have typically different program semantics.

\begin{proposition}
The exchange law for partial strings is not an isomorphism.
\end{proposition}
\begin{proof}
Let  $u$, $v$, $x$ and $y$ be partial strings that only consist of a single event $e_u$, $e_v$, $e_x$ and $e_y$, respectively. Then $(u ; x) \parallel (v ;  y)$ is the following partial string:
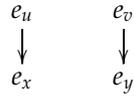
\begin{equation*}
  \xymatrix@R=1.2em{
    e_u\ar@{->}[d] & e_v\ar@{->}[d] \\
    e_x            & e_y
  }
\end{equation*}
In contrast, $(u \parallel v) ;  (x \parallel y)$ is the following partial string:
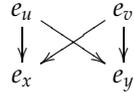
\begin{equation*}
  \xymatrix@R=1.2em{
    e_u\ar@{->}[d]\ar@{->}[dr] & e_v\ar@{->}[d]\ar@{->}[dl] \\
    e_x                        & e_y
  }
\end{equation*}
It is now easy to see that $(u \parallel v) ; (x \parallel y) \not\cong (u ; x) \parallel (v ; y)$.
\end{proof}

The following corollary is directly Lemma 6.8~in~\cite{HMSW2011}:

\begin{corollary}
\label{corollary:exchange-law}
For all $x,y,z \in \sP$, the following inequalities hold:
\begin{itemize}
\item $(x \parallel y); z\sqsubseteq x \parallel (y ;  z)$,
\item $x ;  (y \parallel z) \sqsubseteq (x ;  y) \parallel z$.
\end{itemize}
\end{corollary}
\begin{proof}
By Proposition~\ref{proposition:partial-string-identity}~and~\ref{proposition:exchange-law}, $(x \parallel y) ; z \cong (x \parallel y) ; (\bot \parallel z) \sqsubseteq (x; \bot) \parallel (y; z) \cong x \parallel (y; z)$. An analogous argument proves $x ;  y \parallel z \sqsubseteq x ;  (y \parallel z)$.
\end{proof}

It is not difficult to see that Corollary~\ref{corollary:exchange-law} implies Proposition~\ref{proposition:partial-string-operator-chain}. We have nevertheless given the direct proof of the latter because it served as an introductory exemplar for more complicated proofs about the refinement ordering between partial strings in terms of monotonic bijective morphisms.

This ends this section on partial strings. We next consider the algebraic properties of a family of partial strings.

\section{Programs}
\label{section:programs}

So far we have considered individual partial strings. This is about to change as we embark on the study of programs. In fact, programs are motivation to study \emph{sets} of partial strings as pioneered by Gischer~\cite{G1985} because concurrent programs emit two kinds of different nondeterminism that cannot be modelled by a single partial string alone. To see this, consider the simple program \texttt{if * then P else Q}. If the semantics of a program was a single partial string, then we need to find exactly one partial string that represents the fact \texttt{P} executes or \texttt{Q} executes, but never both. However, a single partial string is not expressive enough for this. Thus, we resort to sets, in fact \emph{downward-closed} sets with respect to $\sqsubseteq$ (Definition~\ref{def:partial-string-isomorphism}) as explained shortly.

\begin{definition}
\label{def:program}
Define a \defn{program} to be a downward-closed set of finite partial strings with respect to $\sqsubseteq$; equivalently $\cX \subseteq \sP_f$ is a program if $\downarrow_\sqsubseteq \cX = \cX$ where $\downarrow_\sqsubseteq \cX \deq \set{y \in \sP_f \alt \exists x \in \cX \colon y \sqsubseteq x}$. Denote with $\bbP$ the family of programs.
\end{definition}

The intuition is that each partial string in a program describes one of possibly many control flows including concurrently executing instructions. Semantically, the downward-closure corresponds to the set of all potential implementations \emph{per} run through the system, explaining the existential quantifier. In other words --- as with Gischer's subsumption order~\cite{G1986} --- the downward-closure over-approximates the behaviour of a concurrent system. And since we only consider systems that terminate, each partial string in this over-approximation is in fact finite.

The transition to powerset of partial strings induces more algebraic laws. Before we study these, we specifically designate the following two sets of finite partial strings:

\begin{definition}
\label{def:program-zero-one}
Define $0 \deq \emptyset$ and $1 \deq \set{\bot}$ where $\bot$ is the empty partial string.
\end{definition}

It is easy to see that both $0$ and $1$ are closed under the downward-closure with respect to $\sqsubseteq$, and therefore they are programs (Definition~\ref{def:program}).

\begin{proposition}
$\downarrow_\sqsubseteq 0 = 0$ and $\downarrow_\sqsubseteq 1 = 1$ in $\bbP$. So $0$ and $1$ are programs.
\end{proposition}
\begin{proof}
Clearly $\downarrow_\sqsubseteq 0 = \downarrow_\sqsubseteq \emptyset = \emptyset = 0$. By Definition~\ref{def:program-zero-one}~and~\ref{def:program}, $1 = \set{y \in \sP_f \alt y \sqsubseteq \bot}$. Let $y \in \sP_f$ such that $y \sqsubseteq \bot$. By Definition~\ref{def:partial-string} of the empty partial string and Definition~\ref{def:partial-string-isomorphism}, there exists a bijective monotonic morphism $f \colon \bot \to y$. By Proposition~\ref{proposition:unique-empty-partial-string}, $y = \bot$. We conclude that $\downarrow_\sqsubseteq 1 = 1$.
\end{proof}

Since $\bbP$ is the family of of $\sqsubseteq$-downward-closed sets of finite partial strings, $\bbP$ clearly forms a complete lattice ordered by subset inclusion where meet, denoted by $\cap$, corresponds to set intersection and join, denoted by $\cup$, is set union with the empty set as bottom, $0$, and the whole of $\sP_f$ as top, i.e. $\top \deq \sP_f$.

\begin{proposition}
\label{proposition:program-lattice}
$\tuple{\bbP, \subseteq, \cap, \cup, 0, \top}$ forms a complete lattice.
\end{proposition}

In the theory of denotational semantics, our complete lattice of programs is called a \emph{Hoare powerdomain} where the ordering $\cP \subseteq \cQ$ for programs $\cP$ and $\cQ$ says that $\cP$ is more deterministic than $\cQ$, or that $\cP$ \emph{refines} $\cQ$. Semantically, the operator $\cP \cup \cQ$ could therefore be seen as the nondeterministic choice of either $\cP$ or $\cQ$. It follows that the equality of programs $\cP$ and $\cQ$ is the same as their subset inclusion both ways, i.e. $\cP = \cQ$ is equivalent to $\cP \subseteq \cQ$ and $\cQ \subseteq \cP$. This means the following holds:

\begin{proposition}
\label{proposition:powerdomain}
For all $\cX, \cY \in \bbP$, $\cX \subseteq \cY$ exactly if $\forall x \in \cX \colon \exists y \in \cY \colon x \sqsubseteq y$.
\end{proposition}
\begin{proof}
Assume $\cX \subseteq \cY$. Let $x \in \cX$. By assumption, $x \in \cY$. By reflexivity of $\sqsubseteq$ (Proposition~\ref{proposition:sqsubseteq-preorder}), $x \sqsubseteq x$. Thus, $\forall x \in \cX \colon \exists y \in \cY \colon x \sqsubseteq y$.

Conversely, assume $\forall x \in \cX \colon \exists y \in \cY \colon x \sqsubseteq y$. Let $x \in \cX$. By assumption, there exists $y \in \cY$ such that $x \sqsubseteq y$. Since programs are downward-closed sets with respect to $\sqsubseteq$ (Definition~\ref{def:program}), $x \in \cY$. Thus, $\cX \subseteq \cY$.
\end{proof}

Next we lift concurrent and sequential composition of partial strings (Definition~\ref{def:bowtie}) to programs as follows:

\begin{definition}
\label{def:program-operator}
For every program $\cX$ and $\cY$, and partial string operator $\Join$, define $\cX \Join \cY \deq\ \downarrow_\sqsubseteq \set{x \Join y \alt x \in \cX \text{ and } y \in \cY}$ where $\cX \parallel \cY$ and $\cX ; \cY$ are called \defn{concurrent} and \defn{sequential program composition}, respectively.
\end{definition}

The meaning of program composition operators derive from their counterpart in the partial string model from Section~\ref{section:partial-strings}. As in the case of partial string operators, we use the ``bow tie'' operator, denoted by $\Join$, as placeholder for either concurrent or sequential program composition, cf.~Definition~\ref{def:bowtie}.

\begin{proposition}[Annihilator]
\label{proposition:program-operator-zero}
For every program $\cX \in \bbP$, $\cX \Join 0 = 0 \Join \cX = 0$.
\end{proposition}
\begin{proof}
Immediate from Definition~\ref{def:program-zero-one}~and~\ref{def:program-operator}.
\end{proof}

Equivalently, by our ``bow tie'' convention, we get the following two statements: $\cX \parallel 0 = 0 \parallel \cX = 0$ and $\cX ; 0 = 0 ; \cX = 0$ for every program $\cX$.

\begin{proposition}[Identity]
\label{proposition:program-operator-one}
For every program $\cX \in \bbP$, $\cX \Join 1 = 1 \Join \cX = \cX$.
\end{proposition}
\begin{proof}
Immediate from Proposition~\ref{proposition:partial-string-identity}.
\end{proof}

\begin{proposition}[Distributivity]
\label{proposition:program-composition-distribute-arbitrary-join}
For every program $\cX$ and family of programs $\set{\cY_n \alt n \in \nats}$, the following equalities holds:
\begin{align*}
\cX \Join \left(\bigcup_{n \geq 0} \cY_n\right) &= \bigcup_{n \geq 0}(\cX \Join \cY_n) \\
\left(\bigcup_{n \geq 0} \cY_n\right) \Join \cX &= \bigcup_{n \geq 0} (\cY_n \Join \cX)
\end{align*}
\end{proposition} 
\begin{proof}
We prove the claim by showing set inclusion both ways. Let $z \in \sP_f$. By expanding definitions, we get the following equivalences:
\begin{align*}
&z \in \cX \Join \left(\bigcup_{n \geq 0} \cY_n\right)\\
&\bimp z \in\ \downarrow_\sqsubseteq \set{x \Join y \alt x \in \cX \land y \in \bigcup_{n \geq 0} \cY_n} \ptag{Definition~\ref{def:program-operator}} \\
&\bimp \exists x, y \in \sP_f \colon x \in \cX \land y \in \bigcup_{n \geq 0} \cY_n \land z \sqsubseteq x \Join y \ptag{Definition of $\downarrow_\sqsubseteq$} \\
&\bimp \exists n \geq 0 \colon \exists x, y \in \sP_f \colon x \in \cX \land y \in \cY_n \land z \sqsubseteq x \Join y \ptag{Definition of set union} \\
&\bimp \exists n \geq 0 \colon z \in\ \downarrow_\sqsubseteq \set{x \Join y \alt x \in \cX \land y \in \cY_n} \ptag{Definition of $\downarrow_\sqsubseteq$} \\
&\bimp \exists n \geq 0 \colon z \in \cX \Join \cY_n \ptag{Definition~\ref{def:program-operator}} \\
&\bimp z \in \bigcup_{n \geq 0}(\cX \Join \cY_n) \ptag{Definition of set union}
\end{align*}
An analogous argument proves the second equation.
\end{proof}

As a corollary, it follows that the sequential and concurrent program composition operators are monotonic in both their arguments.

\begin{corollary}[Monotonicity]
\label{corollary:program-composition-monotonicity}
For every program $\cX$, $\cY$ and $\cZ$ in $\bbP$, if $\cX \subseteq \cY$, then $\cX \Join \cZ \subseteq \cY \Join \cZ$ and $\cZ \Join \cX \subseteq \cZ \Join \cY$.
\end{corollary}
\begin{proof}
Assume $\cX \subseteq \cY$. Equivalently, $\cX \cup \cY = \cY$. Thus, by Proposition~\ref{proposition:program-composition-distribute-arbitrary-join}, $(\cZ \Join \cX) \cup (\cZ \Join \cY) \subseteq \cZ \Join (\cX \cup \cY) = \cZ \Join \cY$. Hence, $\cZ \Join \cX \subseteq \cZ \Join \cY$, proving that concurrent and sequential composition are monotonic in their first argument. Similarly, for the second argument.
\end{proof}

By the Knaster-Tarski fixed point theorem and the fact that $\bbP$ is a complete lattice, it follows that `iterative' concurrent and sequential composition of programs have a (necessarily unique) least fixed point solution. We denote this fixed point by two forms of Kleene star operators.

\begin{proposition}
\label{proposition:program-least-fixed-point}
Let $\cP$ be a program and $F_\cP \colon \bbP \to \bbP$ such that, for all $\cX \in \bbP$, $F_\cP(\cX) = 1 \cup (\cP \Join \cX)$. Then $F_\cP$ has a least fixed point that we denote by $\cP^{\Join}$.
\end{proposition}
\begin{proof}
The conclusion follows from Proposition~\ref{proposition:program-lattice}, Corollary~\ref{corollary:program-composition-monotonicity} and Knaster-Tarski fixed point theorem.
\end{proof}

Therefore, given any program $\cP$ in $\bbP$ according to Definition~\ref{def:program}, we have that $1 \cup (\cP; \cP^;) = \cP^;$. In addition, for every program $\cQ$ in $\bbP$, if $1 \cup (\cP;\cQ) = \cQ$, then $\cP^; \subseteq \cQ$, and similarly for $\parallel$. By the Kleene fixed point theorem, $\cP^\parallel$ and $\cP^;$ for a program $\cP$ can be computed as follows:

\begin{proposition}
\label{proposition:compute-program-least-fixed-point}
Let $\cP$ be a program in $\bbP$. Let $F_\cP \colon \bbP \to \bbP$ be a function such that, for all programs $\cX \in \bbP$, $F_\cP(\cX) = 1 \cup (\cP \Join \cX)$. Then $\cP^{\Join} = \bigcup_{n \geq 1} F^n_\cP(0)$ where $0 \in \bbP$ and $F^1_\cP \deq F$ and $F^{j+1}_\cP \deq F_\cP \circ F^j_\cP$ for all $j \in \nats$.
\end{proposition}
\begin{proof}
By Proposition~\ref{proposition:program-composition-distribute-arbitrary-join}, $F_\cP$ is continuous. The conclusion follows from the Kleene fixed point theorem.
\end{proof}

The following result uses the transitivity of $\sqsubseteq$ (Proposition~\ref{proposition:sqsubseteq-preorder}) to make clear the connection between partial string and program operators. This is key to transfer our knowledge about partial strings to programs.

\begin{lemma}
\label{lemma:connection}
For all programs $\cU, \cV, \cX, \cY \in \bbP$ and pairs of binary operators $\Join_a$ and $\Join_b$ in $\set{; , \parallel}$, if $\forall x, y \in \sP_f \colon x \Join_a y \sqsubseteq x \Join_b y$, then $ \cX \Join_a \cY \subseteq \cX \Join_b \cY$; similarly, if $\forall u, v, x, y \in \sP_f \colon (u \Join_b v) \Join_a (x \Join_b y) \sqsubseteq (u \Join_a x) \Join_b (v \Join_a y)$ and $\Join_a$ is monotonic, then $(\cU \Join_b \cV) \Join_a (\cX \Join_b \cY) \subseteq (\cU \Join_a \cX) \Join_b (\cV \Join_a \cY)$.
\end{lemma}
\begin{proof}
Let $z \in \sP_f$. Assume $z \in \cX \Join_a \cY$. By Definition~\ref{def:program-operator}, $\cX \Join_a \cY =\ \downarrow_\sqsubseteq \set{x \Join_a y \alt x \in \cX \land y \in \cY}$. So there exists $x \in \cX$ and $y \in \cY$ such that $z \sqsubseteq x \Join_a y$. By hypothesis and transitivity of $\sqsubseteq$ (Proposition~\ref{proposition:sqsubseteq-preorder}), there exists $x \in \cX$ and $y \in \cY$ such that $z \sqsubseteq x \Join_b y$. Thus $z \in \cX \Join_b \cY$ because $\cX \Join_b \cY =\ \downarrow_\sqsubseteq (\cX \Join_b \cY)$, i.e. $\cX \Join_b \cY$ is a program. Since $z$ is arbitrary, $\cX \Join_a \cY \subseteq \cX \Join_b \cY$.

The proof of the second implication is analogous except that it also uses the monotonicity of $\Join_a$ on partial strings.
\end{proof}

From the first implication of Lemma~\ref{lemma:connection} and Proposition~\ref{proposition:partial-string-operator-chain} follows that the chain of composition operators carries over to $\pair{\bbP}{\subseteq}$. The next proposition is the first of three frame laws~\cite{HvSMSVZOH2014}.

\begin{proposition}[Frame I]
\label{proposition:frame-i}
For all programs $\cX, \cY \in \bbP$,  $\cX ; \cY \subseteq \cX \parallel \cY$.
\end{proposition}
\begin{proof}
The conclusion follows from Proposition~\ref{proposition:partial-string-operator-chain} and Lemma~\ref{lemma:connection}.
\end{proof}

Noteworthy, the second implication of Lemma~\ref{lemma:connection}, in turn, has as consequence that the exchange law for partial strings (Proposition~\ref{proposition:exchange-law}) generalizes to program composition operators.

\begin{proposition}
\label{proposition:program-exchange-law}
For all $\cU, \cV, \cX, \cY \in \bbP$, $(\cU \parallel \cV) ; (\cX \parallel \cY) \subseteq (\cU ; \cX) \parallel (\cV ; \cY)$.
\end{proposition}
\begin{proof}
By Proposition~\ref{proposition:exchange-law} and the second implication of Lemma~\ref{lemma:connection}.
\end{proof}

\begin{lemma}
\label{lemma:sqsubseteq-inherit-properties}
Let $x, y, z, p \in \sP_f$. If $x \Join y \cong y \Join x$, then $p \sqsubseteq x \Join y$ is equivalent to $p \sqsubseteq y \Join x$. Similarly, if $x \Join (y \Join z) \cong (x \Join y) \Join z$, then $p \sqsubseteq x \Join (y \Join z)$ is equivalent to $p \sqsubseteq (x \Join y) \Join z$.
\end{lemma}
\begin{proof}
The first hypothesis is $x \Join y \cong y \Join x$. Assume $p \sqsubseteq x \Join y$. By hypothesis and Proposition~\ref{proposition:converse-sqsubseteq-partial-order}, $x \Join y \sqsubseteq y \Join x$. By transitivity of $\sqsubseteq$ (Proposition~\ref{proposition:sqsubseteq-preorder}) and assumption, $p \sqsubseteq y \Join x$, proving the forward implication. The backward implication is proved analogously, as well as the second equivalence.
\end{proof}

\begin{lemma}
\label{lemma:program-concurrent-sequential-property}
If $\forall x,y \in \sP_f \colon x \Join y \cong y \Join x$, then $\forall \cX, \cY \in \bbP \colon \cX \Join \cY = \cY \Join \cX$. Similarly, if $\forall x,y \in \sP_f \colon x \Join (y \Join z) \cong (x \Join y) \Join z$, then $\forall \cX, \cY, \cZ \in \bbP \colon \cX \Join (\cY \Join \cZ) = (\cX \Join \cY) \Join \cZ$.
\end{lemma}
\begin{proof}
Let $\cX, \cY \in \bbP$ be programs and $p \in \sP_f$ be a finite partial string. Assume $\forall x,y \in \sP_f \colon x \Join y \cong y \Join x$. We show $\cX \Join \cY \subseteq \cY \Join \cX$ and $\cX \Join \cY \supseteq \cY \Join \cX$ through the following equivalences:
\begin{align*}
&p \in \cX \Join \cY\\
&\bimp p \in\ \downarrow_\sqsubseteq \set{x \Join y \alt x \in \cX \land y \in \cY} \ptag{Definition~\ref{def:program-operator}} \\
&\bimp \exists x, y \in \sP_f \colon x \in \cX \land y \in \cY \land p \sqsubseteq x \Join y \ptag{Definition of $\downarrow_\sqsubseteq$} \\
&\bimp \exists x, y \in \sP_f \colon x \in \cX \land y \in \cY \land p \sqsubseteq y \Join x \ptag{Assumption, commutativity of $\Join$, Lemma~\ref{lemma:sqsubseteq-inherit-properties}} \\
&\bimp \exists x, y \in \sP_f \colon y \in \cY \land x \in \cX \land p \sqsubseteq y \Join x \ptag{Commutativity of conjunction} \\
&\bimp p \in\ \downarrow_\sqsubseteq \set{y \Join x \alt y \in \cY \land x \in \cX} \ptag{Definition of $\downarrow_\sqsubseteq$} \\
&\bimp p \in \cY \Join \cX \ptag{Definition~\ref{def:program-operator}}
\end{align*}

Similarly, for associativity, it suffices to prove subset inclusion both ways which we show in detail to draw attention to the various properties of $\sqsubseteq$ and $\Join$ that are used in the proof:
\begin{align*}
&p \in \cX \Join (\cY \Join \cZ)\\
&\bimp p \in\ \downarrow_\sqsubseteq \set{x \Join q \alt x \in \cX \land q \in (\cY \Join \cZ)} \ptag{Definition~\ref{def:program-operator}} \\
&\bimp \exists x, q \in \sP_f \colon x \in \cX \land q \in (\cY \Join \cZ) \land p \sqsubseteq x \Join q \ptag{Definition of $\downarrow_\sqsubseteq$} \\
&\bimp \exists x, q \in \sP_f \colon x \in \cX \land q \in\ \downarrow_\sqsubseteq \set{y \Join z \alt y \in \cY \land z \in \cZ} \land p \sqsubseteq x \Join q \ptag{Definition~\ref{def:program-operator}} \\
&\bimp \exists x, q \in \sP_f \colon x \in \cX \land (\exists y, z \in \sP_f \colon y \in \cY \land z \in \cZ \land q \sqsubseteq y \Join z) \land p \sqsubseteq x \Join q \ptag{Definition of $\downarrow_\sqsubseteq$} \\
&\bimp \exists x, y, z, q \in \sP_f \colon x \in \cX \land y \in \cY \land z \in \cZ \land q \sqsubseteq y \Join z \land p \sqsubseteq x \Join q \ptag{Definition of $\exists$} \\
&\imp \exists x, y, z, q \in \sP_f \colon x \in \cX \land y \in \cY \land z \in \cZ \land x \Join q \sqsubseteq x \Join (y \Join z) \land p \sqsubseteq x \Join q \ptag{Monotonicity of $\Join$} \\
&\imp \exists x, y, z \in \sP_f \colon x \in \cX \land y \in \cY \land z \in \cZ \land p \sqsubseteq x \Join (y \Join z) \ptag{transitivity of $\sqsubseteq$} \\
&\bimp \exists x, y, z \in \sP_f \colon x \in \cX \land y \in \cY \land z \in \cZ \land p \sqsubseteq (x \Join y) \Join z \ptag{Assumption, associativity of $\Join$, Lemma~\ref{lemma:sqsubseteq-inherit-properties}} \\
&\imp \exists x, y, z, q' \in \sP_f \colon x \in \cX \land y \in \cY \land z \in \cZ \land q' \sqsubseteq x \Join y \land p \sqsubseteq q' \Join z \ptag{Reflexivity and transitivity of $\sqsubseteq$, monotonicity of $\Join$} \\
&\bimp p \in (\cX \Join \cY) \Join \cZ \ptag{Definition~\ref{def:program-operator}}
\end{align*}
An analogous argument proves $\cX \Join (\cY \Join \cZ) \supseteq (\cX \Join \cY) \Join \cZ$, proving associativity of the $\Join$ operator on programs.
\end{proof}

We remark that the previous lemma has as antecedent that a partial string operator is commutative (or associative) only up to isomorphism. In contrast, the consequent of the lemma asserts that programs are in fact \emph{equal} when composed with the corresponding program composition operator. This equality is due to the fact that programs are downward-closed with respect to an ordering which disregards the identity of events.

\begin{definition}
A \defn{semigroup} is an algebraic structure consisting of a set together with an associative binary operation.
\end{definition}

\begin{proposition}
\label{proposition:program-semigroups}
$\tuple{\bbP, \parallel}$ is a commutative semigroup and $\tuple{\bbP, ;}$ is a semigroup.
\end{proposition}
\begin{proof}
Let $\cX, \cY, \cZ \in \bbP$ be programs. By modus ponens, Lemma~\ref{lemma:program-concurrent-sequential-property}~and~Proposition~\ref{proposition:partial-string-associativity}, $\cX \Join (\cY \Join \cZ) = (\cX \Join \cY) \Join \cZ$. Similarly, $\cX \parallel \cY = \cY \parallel \cX$.
\end{proof}

The fact that $\parallel$ is a commutative program operator has the following weak principle of sequential consistency as consequence~\cite{HvSMSVZOH2014}:

\begin{proposition}
For all programs $\cX, \cY \in \bbP$,  $(\cX ; \cY) \cup (\cY ; \cX) \subseteq \cX \parallel \cY$.
\end{proposition}
\begin{proof}
By Proposition~\ref{proposition:frame-i}~and~\ref{proposition:program-semigroups}, $(\cX ; \cY) \subseteq \cX \parallel \cY$ and $(\cY ; \cX) \subseteq \cX \parallel \cY$. Thus, by definition of least upper bound, $(\cX ; \cY) \cup (\cY ; \cX) \subseteq \cX \parallel \cY$. 
\end{proof}

The converse of the previous proposition does not generally hold, a good examplar of the fact that a set of partial strings is more expressive than a set of strings. This link to classical language theory is formalized as follows:

\begin{definition}
\label{def:language}
A \defn{string} is a finite partial string $s$ in $\sP_f$ such that $\preceq_s$ is a total order, i.e. for all $e, e' \in E_s$, $e \preceq_s e'$ or $e' \preceq_s e$. Let $\Gamma^\ast$ be the set of strings. For all programs $\cP$ in $\bbP$, define the \defn{language of $\cP$}, written $\mfrL_\cP$, to be the set of strings where each one refines at least one partial string in $\cP$; equivalently, $\mfrL_\cP \deq \set{s \in \Gamma^\ast \alt \exists p \in \cP \colon s \sqsubseteq p}$.
\end{definition}

In other words, $\mfrL_\cP$ can be seen as the linearizations of all the partial strings in a program according to the refinement ordering of Definition~\ref{def:partial-string-isomorphism}.

\begin{proposition}
\label{proposition:language}
For all programs $\cX, \cY \in \bbP$, if $\cX \subseteq \cY$, then $\mfrL_\cX \subseteq \mfrL_\cY$.
\end{proposition}
\begin{proof}
Immediate from Definition~\ref{def:language}.
\end{proof}

The converse of Proposition~\ref{proposition:language} does not generally hold. For example, $\mfrL_{\cX \parallel \cY} \subseteq \mfrL_{(\cX ; \cY) \cup (\cY ; \cX)}$ but $ \cX \parallel \cY \not\subseteq (\cX ; \cY) \cup (\cY ; \cX)$ for some programs $\cX$ and $\cY$. This shows that our notion of programs from Definition~\ref{def:program} strictly generalizes the concept of sets of strings, i.e. languages consisting of strings.

\begin{definition}
A \defn{monoid} is a semigroup with an identity element.
\end{definition}

\begin{proposition}
\label{proposition:progam-monoid}
$\tuple{\bbP, \parallel, 1}$ is a commutative monoid and $\tuple{\bbP, ;, 1}$ is a monoid.
\end{proposition}
\begin{proof}
Let $\cX \in \bbP$ be a program. By Proposition~\ref{proposition:program-semigroups}, it remains to show that $1$ is the identity for the operator $\Join$ on programs, as shown in Proposition~\ref{proposition:program-operator-one}.
\end{proof}

Of particular interest are complete lattices under the natural order (formally, $x \leq y \deq x \vee y = y$ where $\vee$ denotes join) in which a binary operator distributes over arbitrary least upper bounds ($\bigvee$).

\begin{definition}
\label{def:quantale}
A \defn{quantale} is a complete lattice $Q$ equipped with a semigroup structure $\tuple{Q, \cdot}$ satisfying both complete distributive laws
\begin{align*}
x \cdot \left(\bigvee S\right) = \bigvee \set{y \cdot x \alt y \in S}\\
\left(\bigvee S\right) \cdot x = \bigvee \set{x \cdot y \alt y \in S}
\end{align*}
where $S \subseteq Q$ and $x \in Q$. If $Q$ is a quantale whose semigroup is also a monoid structure $\tuple{Q, \cdot, 1}$, then $Q$ is called a \defn{unital quantale}.
\end{definition}

Quantales appear in different guises in computer science. For example, the powerset of strings over alphabet $\Gamma$ is a unital quantale $\tuple{\powerset{\Gamma^\ast}, \subseteq, \cup, 0, 1, \cdot}$ where $\subseteq$ is the inclusion order on sets, $\cup$ is set union, $0 \deq \emptyset$ is the empty set, $1 \deq \set{\varepsilon}$ is the singleton set with the empty string $\varepsilon$, and the concatenation of two sets of strings, $A$ and $B$, is defined pair-wise, i.e. $A \cdot B \deq \set{a \cdot b \alt a \in A \land b \in B}$, whence $\tuple{\powerset{\Gamma^\ast}, 1, \cdot}$ forms a monoid. Based on the notion of quantales, we succinctly characterize the complete lattice of programs as follows:

\begin{proposition}
\label{proposition:quantale-program-algebra}
The algebraic structure $\tuple{\bbP, \subseteq, \cup, 0, 1, ;, \parallel}$ consists of two unital quantales with respect to sequential and concurrent program composition operators, respectively, that together satisfy the exchange law (Proposition~\ref{proposition:program-exchange-law}).
\end{proposition}
\begin{proof}
By Proposition~\ref{proposition:program-lattice},~\ref{proposition:program-composition-distribute-arbitrary-join}~and~\ref{proposition:progam-monoid}, $\tuple{\bbP, \subseteq, \parallel}$ and $\tuple{\bbP, \subseteq, ;}$ form two quantales according to Definition~\ref{def:quantale}.
\end{proof}

\begin{proposition}
\label{proposition:program-plus}
For all $\cX, \cY, \cZ \in \bbP$, the following holds:
\begin{align*}
&\cX \cup (\cY \cup \cZ) = (\cX \cup \cY) \cup \cZ \\
&\cX \cup \cY = \cY \cup \cX \\
&\cX \cup \cX = \cX \\
&\cX \cup 0 = 0 \cup \cX = 0
\end{align*}
\end{proposition}
\begin{proof}
All these equalities are true since $\cup$ is the least upper bound in the complete lattice $\bbP$ whose bottom element is $0$.
\end{proof}

The next proposition shows that the set of programs forms an algebraic structure commonly known as either an idempotent semiring or dioid. A familiar example of an idempotent semiring is the Boolean semiring $\tuple{\bbB, +, \cdot, 0, 1}$ where $\bbB$ should be interpreted as Boolean values whose sum and product operators are logical disjunction and conjunction, respectively.

\begin{proposition}
\label{proposition:program-semiring}
The algebraic structures $\tuple{\bbP, \cup, \parallel, 0, 1}$ and $\tuple{\bbP, \cup, ;, 0, 1}$ are idempotent semirings where $\tuple{\bbP, \cup, 1}$ is a commutative idempotent monoid, $\tuple{\bbP, \parallel, 1}$ is a commutative monoid and $\tuple{\bbP, ;, 1}$ is a monoid.
\end{proposition}
\begin{proof}
By Proposition~\ref{proposition:program-operator-zero},~\ref{proposition:progam-monoid}~and~\ref{proposition:program-plus}.
\end{proof}

This leads to the main and final result in this section that gives Kleene star operators for concurrent ($\parallel$) and sequential ($;$) program compositions as least fixed points ($\mu$) where the binary join operator ($\cup$) can be interpreted as the nondeterministic choice of two programs.

\begin{theorem}
\label{theorem:program-algebra}
The structure $\mfrS = \tuple{\bbP, \subseteq, \cup, 0, 1, ;, \parallel}$, called \defn{program algebra}, is a complete lattice, ordered by subset inclusion, such that $\parallel$ and $;$ form unital quantales over $\cup$, and $\mfrS$ satisfies the following algebraic laws:
\begin{align*}
& \cX \subseteq \cY \text{ exactly if } \cX \cup \cY = \cY &                                                                                      \\
&(\cU \parallel \cV) ; (\cX \parallel \cY) \subseteq (\cU ; \cX) \parallel (\cV ; \cY) &\quad  \cX \cup (\cY \cup \cZ) = (\cX \cup \cY) \cup \cZ  \\
&\cX \cup \cX = \cX &\quad  \cX \cup 0 = 0 \cup \cX = \cX                                                                                         \\
&\cX \cup \cY = \cY \cup \cX &\quad \cX \parallel \cY = \cY \parallel \cX                                                                         \\
&\cX \parallel 1 = 1 \parallel \cX = \cX &\quad \cX ; 1 = 1 ; \cX = \cX                                                                           \\
&\cX \parallel 0 = 0 \parallel \cX = 0   &\quad \cX ; 0 = 0 ; \cX = 0                                                                             \\
&\cX \parallel (\cY \cup \cZ) = (\cX \parallel \cY) \cup (\cX \parallel \cZ) &\quad  \cX ; (\cY \cup \cZ) = (\cX ; \cY) \cup (\cX ; \cZ)          \\
&(\cX \cup \cY) \parallel \cZ = (\cX \parallel \cZ) \cup (\cY \parallel \cZ) &\quad (\cX \cup \cY) ; \cZ = (\cX ; \cZ) \cup (\cY ; \cZ)           \\
&\cX \parallel (\cY \parallel \cZ) = (\cX \parallel \cY) \parallel \cZ &\quad \cX ; (\cY ; \cZ) = (\cX ; \cY) ; \cZ                               \\
&\cP^{\parallel} = \mu \cX .1 \cup (\cP \parallel \cX) &\quad \cP^{;} = \mu \cX .1 \cup (\cP ; \cX).                                   
\end{align*}
\end{theorem}
\begin{proof}
By Proposition~\ref{proposition:compute-program-least-fixed-point},~\ref{proposition:program-exchange-law},~\ref{proposition:program-semiring}~and~\ref{proposition:quantale-program-algebra}.
\end{proof}

Using the terminology established by Tony Hoare et al.~\cite[p. 274]{HMSW2011}, we can summarize $\tuple{\bbP, \subseteq, \cup, 0, 1, ;, \parallel}$ as a partial order model of a concurrent quantale. Furthermore, the fact that the program algebra from Theorem~\ref{theorem:program-algebra} is a complete lattice makes it a CKA~\cite[p. 276]{HMSW2011}. This concludes our construction.

\section{Concluding remarks}

This paper gives the technical details for adapting Gischer's pomset model and \'{E}sik's monotonic bijective morphisms to construct a partial order model of computation that satisfies the axioms of a recently developed algebraic semantics by Tony Hoare and collaborators. The constructions in this paper are particularly guided by the problem of symbolically encoding concurrency into quantifier-free first-order logic formulas. This outlook has an algorithmic flavour to it that is relevant for automated proof techniques and computer-aided formal verification of concurrent systems. In subsequent work we will show how to use the partial string model proposed in this paper to give the first symbolic refinement checking algorithm of concurrent systems for relaxed memory.

\bibliographystyle{plain}
\bibliography{partial-strings}

\end{document}